\documentclass[11pt,oneside]{article}
\usepackage{etoolbox}
\usepackage{amsmath}
\usepackage{amsfonts}
\usepackage{amsthm}
\usepackage{enumerate}
\usepackage{physics}
\usepackage[T1]{fontenc}
\usepackage{mathrsfs}
\usepackage{mathdots}
\usepackage{graphicx}
\usepackage{mathtools}
\graphicspath{ {./images/} }
\usepackage{xcolor}
\usepackage{hyperref}
\hypersetup{colorlinks=true,linkcolor=red,citecolor=blue,filecolor=magenta,urlcolor=cyan}
\usepackage{subfiles}
\usepackage{pgfplots} 
\usepackage{tikz}
\usepackage{algpseudocode}
\usepackage{algorithm}
\usepackage[font=small, margin=0.2\textwidth]{caption}

\usepackage[utf8]{inputenc}
\newtheorem{theorem}{Theorem}[section]

\newtheorem{lemma}[theorem]{Lemma}
\newtheorem{proposition}[theorem]{Proposition}

\theoremstyle{definition}

\newtheorem{definition}[theorem]{Definition}

\newcommand{\Z}{\mathbb{Z}}

\numberwithin{equation}{section}

\newtheorem{innercustomgeneric}{\customgenericname}
\providecommand{\customgenericname}{}
\newcommand{\newcustomtheorem}[2]{%
  \newenvironment{#1}[1]
  {%
   \renewcommand\customgenericname{#2}%
   \renewcommand\theinnercustomgeneric{##1}%
   \innercustomgeneric
  }
  {\endinnercustomgeneric}
}

\newcustomtheorem{customthm}{Theorem}
\newcustomtheorem{customlemma}{Lemma}
\newcustomtheorem{customdef}{Definition}
\newcustomtheorem{customcor}{Corollary}
 
\begin{document}

\title{Quantum and Classical Algorithms for Bounded Distance Decoding}
 
%  \author{Richard Allen \thanks{\href{mailto:rallen@college.harvard.edu}{rallen@college.harvard.edu}}
%         \and
%         Ratip Emin Berker \thanks{\href{mailto:rberker@college.harvard.edu}{rberker@college.harvard.edu}}
%         \and
%         S\'ilvia Casacuberta \thanks{\href{mailto:scasacubertapuig@college.harvard.edu}{scasacubertapuig@college.harvard.edu}}
%         \and
%         Michael Gul \thanks{\href{mailto:michaelgul@college.harvard.edu}{michaelgul@college.harvard.edu}}}

 \title{Quantum and Classical Algorithms for Bounded Distance Decoding}
 
 \author{Richard Allen\thanks{\href{mailto:rallen@college.harvard.edu}{rallen@college.harvard.edu}}
        \and
        Ratip Emin Berker\thanks{\href{mailto:rberker@college.harvard.edu}{rberker@college.harvard.edu}}
        \and
        S{\'i}lvia Casacuberta\thanks{\href{mailto:scasacubertapuig@college.harvard.edu}{scasacubertapuig@college.harvard.edu}}
        \and
        Michael Gul\thanks{\href{mailto:michaelgul@college.harvard.edu}{michaelgul@college.harvard.edu}}}
\date{Harvard University \\[2ex]
February 17, 2022}
\maketitle

\abstract{In this paper, we provide a comprehensive overview of a recent debate over the quantum versus classical solvability of bounded distance decoding (BDD). Specifically, we review the work of Eldar and Hallgren \cite{eh21}, \cite{h21} demonstrating a quantum algorithm solving $\lambda_1 2^{-\Omega(\sqrt{k \log q})}$-BDD in polynomial time for lattices of periodicity $q$, finite group rank $k$, and shortest lattice vector length $\lambda_1$. Subsequently, we prove the results of \cite{dvw21}, \cite{dvw21v2} with far greater detail and elaboration than in the original work. Namely, we show that there exists a deterministic, classical algorithm achieving the same result.}

{
  \hypersetup{linkcolor=black}
  \tableofcontents
} 
 
\newpage
 
\section{Introduction}

 Lattice-based encryption, and specifically encryptions based on the Learning With Errors (LWE) problem as introduced by Regev~\cite{reg05}, are considered a promising candidate for post-quantum encryption. The security of LWE rests on the conjectured hardness of lattice problems such as the Closest Vector Problem (CVP) and Bounded Distance Decoding (BDD). However, Eldar and Hallgren \cite{h21} have recently proposed a quantum algorithm for BDD with a subexponential approximation factor.
 The specific quantum protocol used to achieve this exponential speed-up raised questions on the purely classical achievability of the same result. Specifically, Eldar and Hallgren's algorithm uses the Quantum Phase Estimation (QPE) protocol to reduce a worst-case instance of BDD to a random instance in a lattice of lower dimension, while preserving the bound on error magnitude. The resulting problem has a known classical solution. Since the core quantum step is the use of QPE for dimensionality reduction, a natural question is whether classical dimensionality reduction algorithms could be substituted. Shortly after this result, Ducas and van Woerden \cite{dvw21} showed that the Lenstra-Lenstra-Lov\'asz (LLL) lattice basis reduction algorithm can be used to achieve the same result for a specific class of lattices, later generalizing this result to the same type of lattice considered by Eldar and Hallgren.
 In this work, we give an overview of Eldar and Hallgren's quantum algorithm, followed by a detailed exposition of Ducas and van Woerden's classical algorithm. In order to develop the necessary tools for the proof, we survey the field of lattice problems in general, and in particular describe the LLL algorithm in detail. 

 \subsection{Timeline of Publications}
 Eldar and Hallgren first presented their algorithm in a talk given by Hallgren on September 21, 2021 at the Simons Institute \cite{h21}. Their main point can be summarized as follows:
 \begin{theorem}[Eldar and Hallgren \cite{eh21}]\label{OHAyani}
    There exists a $\textnormal{poly}(n, log \ q)$-time quantum algorithm solving $\lambda_1 2^{-\Omega(\sqrt{k \log q})}$-BDD on lattices of dimension $n$, periodicity $q$, and finite group rank $k$ with shortest lattice vector length $\lambda_1$. 
\end{theorem}
We elaborate on the definitions and implications of the parameters $n,q,$ and $k$ further below. While Hallgren's talk only considered the special case $q = 2^n, k = 1$, it was also briefly mentioned that the theorem holds for arbitrary $q$ and $k$. Three days after the talk, on September 24, 2021, Ducas and van Woerden published a note \cite{dvw21} proving that known classical polynomial-time algorithms (namely, LLL and Babai's algorithms) are sufficient to achieve the result presented by Eldar and Hallgren. Similar to Hallgren's talk, the note only considered the special case where $q = 2^n, k = 1$, hence presenting the theorem:
\begin{theorem}[Ducas and van Woerden, Version 1 \cite{dvw21}]
For any given vector $a \in \mathbb{Z}^n$, define the lattice $L_a=q\mathbb{Z}^n+a\mathbb{Z}$. There exists a deterministic polynomial-time algorithm that solves BDD in $L_a$ for any error up to radius  $\lambda_1(L_a)\cdot 2^{-\Theta(\sqrt{n})}$, where  $\lambda_1(L_a)$ is the length of the shortest vector in $L_a$.
\end{theorem}

Given that Hallgren had mentioned in his talk that their result is generalizable to arbitrary $q$ and $k$, this motivated us to work on generalizing Ducas and van Woerden's proof to arbitrary $q$ and $k$ as well. Following this, we received Eldar and Hallgren's preprint through private correspondence, in which their result was indeed proven for arbitrary $q$ and $k$ \cite{eh21}. While we were working independently on a proof of the classical algorithm for arbitrary $q$ and $k$, we realized later that on October 14, 2021, Ducas and van Woerden had published a second version of their note \cite{dvw21v2}, explaining that their result is generalizable for arbitrary $q$ and $k$:

%\newpage

 \begin{theorem}[Ducas and van Woerden, Version 2 \cite{dvw21v2}]\label{OHAyanishey}
For any given matrix $\mathrm{\mathbf{A}} \in \mathbb{Z}^{n\times k}$, define the lattice $L_\mathrm{\mathbf{A}}=q\mathbb{Z}^n+\mathrm{\mathbf{A}}\mathbb{Z}^k$. There exists a deterministic polynomial-time algorithm that solves BDD in $L_\mathrm{\mathbf{A}}$ for any error up to radius  $\frac{1}{2}\lambda_1(L_\mathrm{\mathbf{A}})\cdot \exp{-\Omega( \sqrt{k \cdot \ln q})}$.
\end{theorem}
Note that in principle, $\mathbf{A}$ is arbitrary, but in order for $k$ to give the finite group rank we generally restrict the columns of $A$ to be linearly independent over $\Z$. From this point on, we sought to elaborate on their result, filling in missing details and unifying a summary of Eldar and Hallgren's algorithm with an elaboration on Ducas and van Woerden's proof for the general case.
 
\subsection{Outline of this Paper}

This paper is structured as follows. In Section~\ref{sec:background}, we introduce the necessary definitions and tools to work with lattices (Section~\ref{sec:21}), as well as the main computational lattice problems that we require for the discussion, mainly SVP, CVP, and BDD (Section~\ref{sec:22}), along with their known hardness results and reductions (Section~\ref{sec:23}). We also provide a survey on other important lattice-based quantum algorithms beyond the result of Eldar and Hallgren (Section~\ref{sec:24}). In Section~\ref{sec:3}, we provide an overview of the quantum algorithm by Eldar and Hallgren, including the necessary quantum primitives (Section~\ref{sec:31}) and the intuition for their proof (Section~\ref{sec:32}). In Section~\ref{sec:classical}, we review the main two classical lattice algorithms, namely the LLL algorithm (Section~\ref{sec:lll}) and Babai's nearest plane algorithm (Section~\ref{sec:babai}). 

As we describe the algorithms and their correctness, we extract the necessary facts for the proof of the classical counterpart to Eldar and Hallgren's quantum algorithm, used in Section~\ref{sec:proof}. Our main proof is split into three parts: we first show how to obtain short bases from full-rank sets of short lattice vectors (Section~\ref{sec:51}), we then turn to dual lattices and their relevant properties (Section~\ref{sec:52}), and lastly we put together all the ingredients of the proof (Section~\ref{sec:lastfacts}).

\section{Background on Lattices, Relations, and Reductions}\label{sec:background}

\subsection{Definitions and Notation}\label{sec:21}

Before turning our attention to BDD, we survey a number of lattice problems and the connections between them. The basic objects of interest are integer lattices:
 
\begin{definition}[Lattice] \label{deneme}
    A lattice $\mathcal{L}$ is generated by all the integer combinations of the vectors of some basis $\textbf{B}$:
    \[
        \mathcal{L} = \sum_{i=1}^m \mathbb{Z}b_i = \Big\{\sum_{i=0}^m z_ib_i, \textrm{ where } z_i \in \mathbb{Z}, b_i \in \mathbf{B} \Big\}.
    \]
\end{definition}

\begin{figure}[h!]
    \centering
    \includegraphics[width=0.7\textwidth]{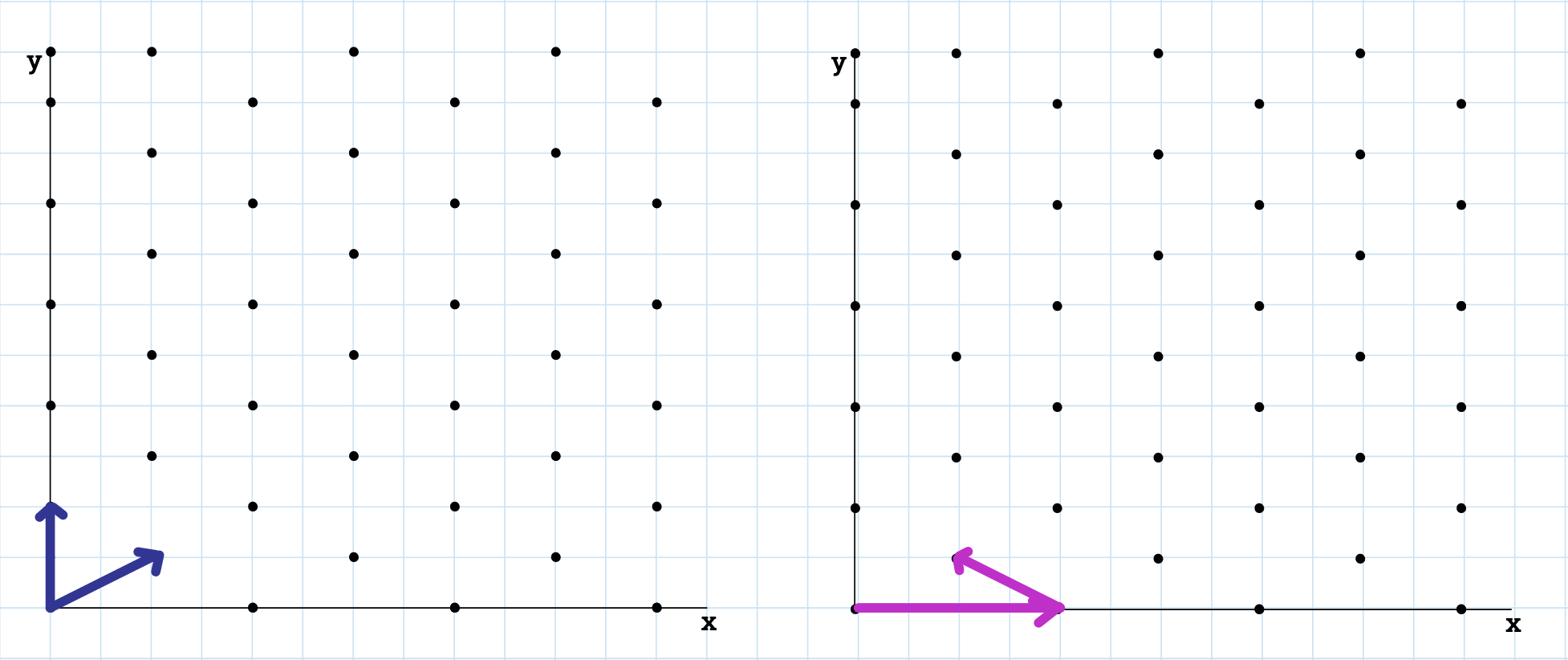}
    \caption{An example of a lattice in $\mathbb{R}^2$, with two bases given which correspond to $\{[2, 1], [0, 2]\}$ and $\{[4, 0], [-2, 1]\}$.}
    \label{fig:my_label}
\end{figure}

Lattices have infinitely many bases, and so it is natural to ask which bases are \textit{better} in some sense. We will see that the two main properties we will consider are that the basis vectors are as \textit{short} and as \textit{orthogonal} as possible. In fact, one of the key insights guiding shortest vector algorithms is that the latter of these properties implies the former. We will discuss this further in Section~\ref{sec:classical}, when we introduce the LLL algorithm.

The volume of the lattice is denoted by $\det(\mathcal{L})$. For any basis $\textbf{B} = \{b_1, \ldots, b_n\}$ for $\mathcal{L}$, the determinant of $\mathcal{L}$ is defined by
\[
    \det(\mathcal{L}) = \sqrt{ \left|\det(\textbf{B}^\top \textbf{B})\right|},
\]
where $\textbf{B}$ is understood as the basis matrix constructed with the vectors $b_i$ as columns. If $\mathcal{L}$ is full-rank, then \textbf{B} is a square matrix, so we can also think of the volume of the lattice as $|\det(\textbf{B})|$. Geometrically, the determinant of the lattice is the volume of the parallelepiped spanned by the $b_i$'s \cite{epfl}. It is important to remark that the volume of the lattice is an invariant; that is, does not depend on the choice of the basis. This is because any two different bases $\textbf{B}_1, \textbf{B}_2$ for the same lattice $\mathcal{L}$ are related by a unimodular matrix \cite{epfl}. Specifically, there exists a matrix $U$ with $\det(U) = \pm 1$ such that
\[
    \textbf{B}_1 = U \cdot \textbf{B}_2.
\]
Hence it follows that $|\det(\textbf{B}_1)| = |\det(\textbf{B}_2)|$.

We will return to the topic of volume of lattices in Section~\ref{sec:classical}, when we further analyze the properties of $\det(\mathcal{L})$.

\subsection{SVP, CVP, and BDD}\label{sec:22}

Long before lattices were used in cryptography, computational problems in the field of lattice theory had been considered in other contexts. The two main ones are the Shortest Vector Problem (SVP) and the Closest Vector Problem (CVP):

\begin{definition}[$\gamma$-SVP \cite{lucaslec}]
    Given a lattice $\mathcal{L} \subset \mathbb{R}^n$, find a lattice vector $x \in \mathcal{L} \setminus {0}$ such that
    \[
        ||x|| \leq \gamma(n) \cdot \lambda_1(\mathcal{L}),
    \]
    where $\lambda_1(\mathcal{L})$ denotes the length of the shortest vector in the lattice.
\end{definition}

\begin{definition}[$\gamma$-CVP \cite{noahcvp}]
    Given a lattice $\mathcal{L} \subset \mathbb{R}^n$ and a target vector $t \in \mathbb{R}^n$, find a lattice vector $x \in \mathcal{L}$ such that
    \[
        ||x-t|| \leq \gamma(n) \cdot \textrm{dist}(\mathcal{L}, t),
    \]
    where $\textrm{dist}(\mathcal{L}, t) = \min_{x \in \mathcal{L}} ||x-t||$. See Figure~\ref{fig:cvp} for a diagram.
\end{definition}

\begin{figure}[h!]
    \centering
    \includegraphics[width=0.3\textwidth]{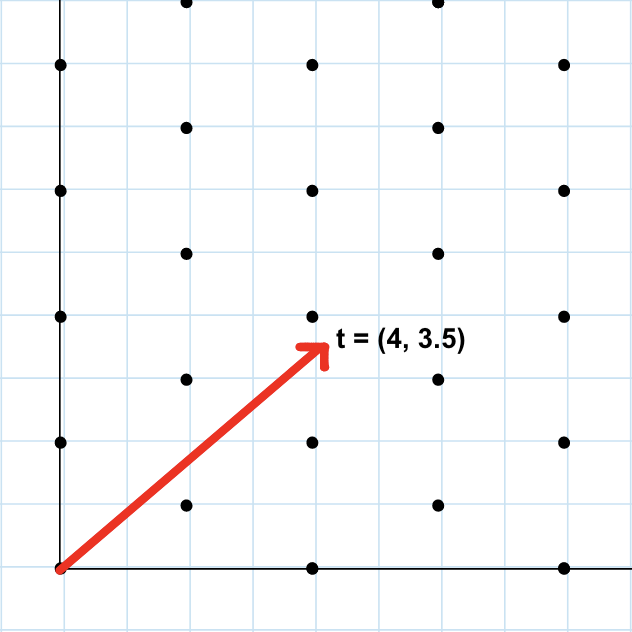}
    \caption{An instance of the CVP, with target vector $t = (4, 3.5)$ and solution $(4, 4)$. Note that the shortest vector in this lattice is $(0, 2)$.}
    \label{fig:cvp}
\end{figure}

These problems have many variants. In addition to the approximation versions presented above, both $\gamma$-SVP and $\gamma$-CVP have exact analogues (called SVP and CVP respectively), where the goal is to find the shortest/closest vector instead of approximating it. And, $\gamma$-SVP also has a decision-based version called $\gamma$-GAP-SVP:
\begin{definition}[$\gamma$-GAP-SVP \cite{reg05}] 
   Given a lattice $\mathcal{L} \subset \mathbb{R}^n$ and $d > 0$, output $\textsf{YES}$ if $\lambda_1(\mathcal{L}) \le d$ and $\textsf{NO}$ if  $\lambda_1(\mathcal{L}) \ge \gamma(n) \cdot d$.
\end{definition} 

A slightly different problem, considered by Eldar and Hallgren \cite{h21}, is Bounded Distance Decoding (BDD), which is concerned with recovering a lattice vector given a target vector with a slight error.

\begin{definition}[Bounded Distance Decoding \cite{dvw21v2}]
Given a lattice $\mathcal{L} \subset \mathbb{R}^n$, a radius $\frac{\lambda_1(\mathcal{L})}{2} > r > 0,$ and a vector $t = v + e$ with $\|e\| < r$ and $v \in \mathcal{L},$ output the lattice vector $v.$
\end{definition}
Note that the condition $r < \frac{\lambda_1(\mathcal{L})}{2}$ guarantees a unique solution. BDD is the main problem that we will be studying in this paper, since both the quantum and classical algorithms discussed below give subexponential approximations. As we will develop in Section~\ref{sec:classical}, there are classical polynomial-time algorithms that can approximate SVP and CVP for $\gamma(n) = 2^{O(n)}$; namely, the LLL algorithm and Babai's nearest plane algorithm. More concretely, the LLL algorithm solves $\gamma$-SVP for $\gamma(n) = 2^{(n-1)/2}$, whereas Babai's nearest plane algorithm solves $\gamma$-CVP for $\gamma(n) = 2^{n/2}$.\footnote{By tuning the parameters of the algorithm, it is actually possible to improve the approximation factors to $\gamma(n) = (2/\sqrt{3})^n$ and $2(2/\sqrt{3})^n$, respectively \cite{regevcvp}. This approximation factor can be further improved with BKZ bases \cite{lll}.} One can then use the same algorithm by Babai to solve BDD whenever $2^{n/2}\cdot dist(\mathcal{L}, t) \leq \lambda_1/2$, given that in that case the answer is unique.

\subsection{Lattice Hardness Results}\label{sec:23}

There are several important intractability results about such computational lattice problems. In 1981, van Emde Boas proved that CVP is NP-hard, but the same result did not immediately extend to SVP \cite{boas}. In fact, $\gamma$-CVP is now known to be NP-complete for any $\gamma \leq n^{c/\log \log n}$. In 1997, Ajtai proved that SVP is NP-hard to solve exactly under randomized reduction \cite{ajtai}, which provided the first theoretical evidence that SVP is intractable. Later on, Micciancio proved that SVP is NP-hard to approximate with any factor less than $\sqrt{2}$ \cite{Micciancionphard}. Lastly, Goldreich, Micciancio, Safra, and Seifer showed that SVP is not harder than CVP \cite{gmss99}. 
Additionally, the various lattice problems are connected by a network of reductions. Figure~\ref{fig:reductions} contains an overview of the major dimension-preserving reductions.   

\begin{figure}[h!]
    \centering
    \includegraphics[width=0.9\textwidth]{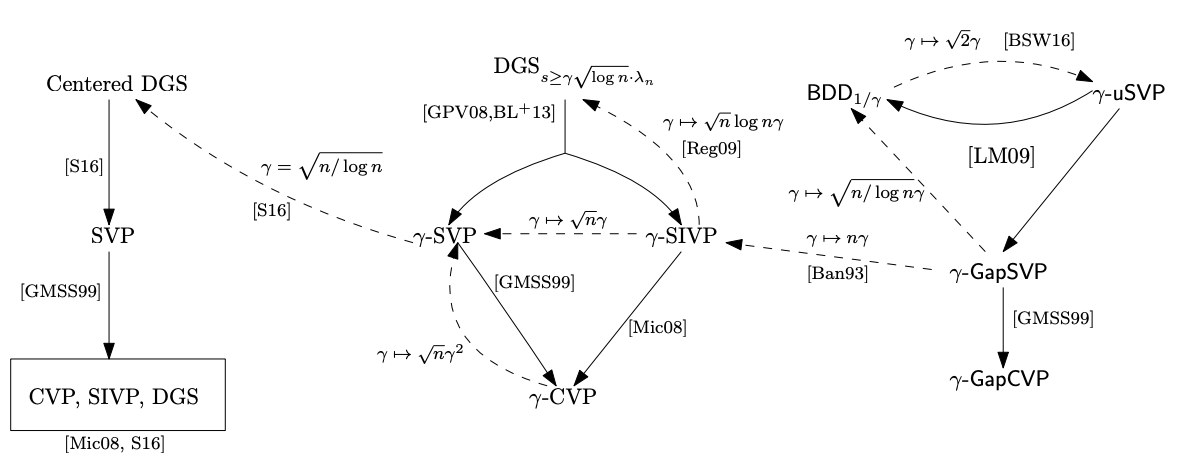}
    \caption{A summary of significant dimension-preserving reductions between various lattice problems, where $A \rightarrow B$ means $A$ reduces to $B$. Solid (dotted) arrows denote a reduction that preserves (changes) the approximation factor $\gamma$. Source: \cite{latticereds}.}
    \label{fig:reductions}
\end{figure}

While interesting in their own right, lattice problems are especially relevant for cryptography. Currently, the best candidates for post-quantum cryptography are lattice-based encryptions, and the security of these schemes relies on the conjectured hardness of these lattice problems. When Learning With Errors (LWE) was first introduced by Regev he proved its security by reducing it to $\gamma$-GAP-SVP: 

\begin{theorem}[Regev, Hardness of Learning With Errors \cite{reg05}]
Define $LWE_{p, \chi}$ as follows: given a series of equations of the form 
$$\langle s, a_i \rangle + e_i = b_i$$
where $s\in \mathbb{Z}_p^n$, the $a_i$ are independently and uniformly chosen from $\mathbb{Z}_p^n,$ $b_i \in \mathbb{Z}_p$, and the $e_i$ are independently chosen from $\chi,$ recover $s.$ Then, if $\alpha \in (0,1)$ is such that $\alpha p > 2\sqrt{n}$, and there is an efficient algorithm that can solve $LWE_{p, \Bar{\Psi}_\alpha},$ then there exists an efficient quantum algorithm that solves $\tilde{O}(n/\alpha)$-GAP-SVP in the worst case. 
\end{theorem}

The distribution $\Bar{\Psi}_\alpha$ is described at depth in the paper, but it is essentially a discrete Gaussian centered at 0 with standard deviation $\alpha p$.

Crucially, $\gamma$-GAP-SVP is conjectured to be hard to solve, i.e., no polynomial time algorithm can solve it for a polynomial $\gamma$. And, Micciancio proved that for $\gamma(n) \le \sqrt{2},$ $\gamma$-GAP-SVP  is NP-Hard \cite{Micciancionphard}. Though Regev's reduction is quantum, Peikert demonstrated a classical reduction in \cite{Peikert09}, and in 2009, Lyubashevsky and Micciancio \cite{LM09} reduced $\gamma$-GAP-SVP to BDD: 
\begin{theorem}[Theorem 7.1 in \cite{LM09}] 
For any $\gamma > 2\sqrt{n / \log{n}}$, there is a polynomial time Cook reduction from $\gamma$-GAP-SVP to $\tfrac{1}{\gamma}\sqrt{n/\log{n}}$-BDD.
\end{theorem}

Elsewhere in the paper, they proved the opposite direction of the reduction as well, meaning that the conjectured hardness of $\gamma$-GAP-SVP is equivalent to the hardness of BDD. It is precisely because of this link between the security of LWE and BDD that further study of Eldar and Hallgren's algorithm for BDD is so important.

\subsection{Survey of Quantum Approaches to Lattice Problems}\label{sec:24}

Before Eldar and Hallgren's algorithm for BDD, a number of other quantum algorithms have been proposed for different lattice problems. For instance, Kuperberg describes in \cite{copernicus} a quantum algorithm utilizing Quantum Fourier Transform (QFT) to solve the dihedral hidden subgroup problem (DHSP) that has time complexity $2^{O(\sqrt{\log N})}$ and takes $2^{O(\sqrt{\log N}) }$ quantum space, where $2N$ is the order of the dihedral group $D_N$. Regev applies a similar quantum algorithm to the unique shortest vector problem (unique SVP) in \cite{reg04} and proposes a space-efficent alternative to Kuperberg's algorithm in \cite{reg06}, taking only $O(\log{N})$ quantum space instead of $2^{O(\sqrt{\log N})}$. Kuperberg then generalizes Regev's algorithm in \cite{copernicus2},  adding two parameters that (1) allows the algorithm to take less space at the cost of more quantum time and (2) allows the algorithm take less quantum time at the cost of more  classical space and time, as long as the the classical space has quantum access. However, the composition of Kuperberg and Regev's algorithm does \textit{not} yield a subexponential algorithm for SVP. The reason is that Regev's reduction produces a quadratic blowup in the input size. 

Overall, these papers are good examples of how quantum algorithms can be used in order to improve classical approaches to lattice-based problems by providing time and space efficient variations, which we imagine was also what motivated Eldar and Hallgren's work.

\section{Overview of Eldar and Hallgren's Algorithm}\label{sec:3}

We now give an overview of some of the key points in Eldar and Hallgren's proof of Theorem~\ref{OHAyani} \cite{eh21}. Eldar and Hallgren parametrize lattices by the following three integers:
\begin{enumerate}
    \item Lattice dimension $n$: $n$ is the dimension of the space $\Z^n$ in which our lattice is embedded. We consider only full rank lattices, i.e., those for which the cardinality of the basis $\mathbf{B}$ is also $n$, since any lattice can be reduced to a full-rank lattice in a lower-dimensional $\Z^m$.
    
    \item Periodicity $q$: $q$ is defined as the minimal integer such that $q \Z^n \subset \mathcal{L}$. In other words, it is the minimal integer such that $q e_i \in \mathcal{L}$ for every standard basis vector $e_i \in \Z^n$. Such an integer must exist. To see why, note that there must be some minimal integer $q_i$ such that $q_i e_i \in \mathcal{L}$, because otherwise we could extend our basis $\textbf{B}$ by $e_i$, which is impossible since $\mathcal{L}$ is full-rank. Then let $q = \text{LCM}(\{q_i\})$.
    
    \item Finite group rank $k$: Since $q \Z^n$ is a subgroup of the abelian group $\mathcal{L}$, we can consider the quotient $\mathcal{L} \text{ mod }q \equiv \mathcal{L} / q \Z^n$, obtained by taking each standard basis coefficient of a lattice vector mod $q$. This is a finite abelian group since $\Z_q$ is finite, so by the structure theorem for finitely generated abelian groups we can decompose $\mathcal{L} \text{ mod }q$ as $\Z_{q_1} \times \cdots \times \Z_{q_k}$ for $q_i \mid q_{i+1}$. Moreover, $k$ is the group rank of $\mathcal{L} \text{ mod }q$ and can be thought of as the dimension of the reduced lattice, analogous to the security parameter in LWE. Letting $\textbf{G}$ give a basis for $\mathcal{L} \text{ mod }q$, it follows that the matrix $[\textbf{G} | q \textbf{I}]$ for an $n \times n$ identity $\textbf{I}$ generates the full lattice $\mathcal{L}$ (but is not a basis; it is too large to be a linearly independent set).
\end{enumerate}
Now that we understand the claim, we begin with a review of some major quantum primitives used in the construction.

\subsection{Quantum Primitives}\label{sec:31}

There are only two quantum primitives necessary to understand Eldar and Hallgren's algorithm, and one is derived from the other. First is the Quantum Fourier Transform (QFT), a particular basis change for quantum states. 
\begin{definition}
    Let $\ket{x}$ for $0 \le x < q$ denote a basis state for some $q$-dimensional Hilbert space $\mathcal{H}$ of quantum states. The Quantum Fourier Transform (QFT) over the cyclic group $\Z_q$, denoted $F_q$, is a linear map defined by
    $$
    F_q: \ket{x} \mapsto \frac{1}{\sqrt{q}} \sum_{i=0}^{q-1} \omega_q^{ix} \ket{i},
    $$
    where $\omega_q$ denotes a primitive $q$th root of unity.
\end{definition}

The QFT is computable in poly(log($q$)) time using only Hadamard and controlled phase rotations. Note that this definition differs from that presented in class and \cite{boaz}, which uses only Hadamard gates and no controlled operations, but conforms to the definition most often used in physics (by leaving out the controlled phase rotations we obtain what physicists refer to as the \textit{Hadamard transformation}). One of the main uses of the QFT is as the basic building block for the Quantum Phase Estimation (QPE) algorithm:
\begin{theorem}[Pg. 225 in \cite{mikeandike}]
    There exists a quantum algorithm which, given a unitary $U$, and an eigenstate $\ket{u}$ of $U$ with eigenvalue $e^{2 \pi i \phi_u}$, outputs an $n$-bit approximation $\tilde{\phi}_u$ to $\phi_u$ with probability $1 - \epsilon$ using $t = n + O(\log(1/\epsilon))$ ancilla qubits and $O(t^2)$ operations.
\end{theorem}

The algorithm works by first applying a series of controlled unitary powers to the ancilla qubits, which transforms our ancilla register into a QFT state. We can then apply the inverse QFT to obtain an approximation of the phase of the unitary eigenvalue.

With these quantum primitives covered, we can understand the basic arguments behind Eldar and Hallgren's algorithm, presented below.

\subsection{Intuition for the Proof}\label{sec:32}

Eldar and Hallgren's algorithm reduces worst-case BDD on a lattice $\mathcal{L}$ in $\mathcal{L}(n,q,k)$ to \textit{$\tilde{\epsilon}$-random-BDD}, which is BDD on a random lattice $\tilde{\mathcal{L}}$ with basis $[\tilde{\textbf{B}} | q \textbf{I}]$ for $\tilde{\textbf{B}} \in \Z_q^{m \times k}$ selected randomly. $\tilde{\mathcal{L}}$ has shortest vector of length $\tilde{\lambda}_1$, and the transformed target vector $\tilde{t}$ is at most $\tilde{\lambda}_1 \tilde{\epsilon}$ away from a lattice vector in $\tilde{\mathcal{L}}$.  For $m = \sqrt{k \log q}$, we can solve BDD on this lattice with approximation factor $\tilde{\epsilon} = 2^{-\sqrt{k \log q}}$ using Babai's nearest plane algorithm.

Before obtaining this random lattice, we consider a related, long-standing problem in quantum state representations of a lattice. For decades, researchers have considered the goal of generating, given a lattice $\mathcal{L}$ with coefficients in $\Z_q$ with basis $\mathbf{B}$, the state $\ket{\psi_0} = \sum_{v \in \mathcal{L}} \sum_{z \in C} \ket{v+z}$ where $C$ is some extended shape, say a cube (more on why this is desirable below). We approximate $\ket{\psi_0}$ in the following way. Begin by computing the state $\sum_c \ket{c} \otimes \sum_{z \in C} \ket{z}$ where $c$ ranges over all possible coefficients of lattice vectors in $\mathcal{L}$. We can then perform entangling gates to obtain $\sum_c \sum_{z \in C} \ket{c, \mathbf{B}c + z}$. We want the state $\sum_c \sum_{z \in C} \ket{0, \mathbf{B}c + z}$, but computing this would require us to know the $c$ corresponding to each $\mathbf{B}c+z$, i.e., to solve BDD, which is exactly the problem we are trying to solve in the first place!

We cannot do this, but we can perform a Quantum Fourier Transform and measure the first register to compute $\ket{\psi_a} = \sum_c \omega_q^{c \cdot a} \sum_{z \in C} \ket{\mathbf{B}c + z}$ for $a$ a random vector (random by the nature of quantum measurement). The problem of computing $c$ from $\mathbf{B}c+z$ is therefore reduced to the problem of determining the phases $\omega_q^{c \cdot a}$. The key observation is that the state $\ket{\psi_a}$ is an eigenvector of a shift by a lattice vector. For $x \in \Z_q^n$, let $U_x$ be the shift unitary defined by $U_x \ket{y} = \ket{y+x}$. Then if we shift by the closest lattice vector $\mathbf{B} c' \in L$ to $t$, we obtain $U_{Bc'} \ket{\psi_a} = \omega_q^{c' \cdot a} \ket{\psi_a}$. Applying the Quantum Phase Estimation (QPE) algorithm, we can compute $c' \cdot a$, and after repeating for multiple random $a$ we can obtain a good approximation of $c'$.

Unfortunately, we are trying to solve for $\mathbf{B} c'$, so we cannot proceed so directly. However, we do have the target vector $t = \mathbf{B} c' + \Delta$ for some $\Delta$ of norm regulated by the BDD approximation factor, and we have $U_t \ket{\psi_a} = \omega_q^{c' \cdot a} U_{\Delta} \ket{\psi_a} \approx \omega_q^{c' \cdot a}
\ket{\psi_a}$. (``Approximately'' here is defined in terms of state overlap. If we had not summed over cubes $C$ about each lattice point, then for any $\Delta$ not a lattice vector this overlap would be zero. This justifies our consideration of these states). This is the key quantum subroutine used in the dimensionality reduction: Eldar and Hallgren use a modified approximate QPE to sample $m = O(\sqrt{k \log q})$ noisy inner products of $c'$ and construct a lattice $\tilde{\mathcal{L}}$ of the form described above. Then they simply use Babai's nearest plane algorithm on $\tilde{\mathcal{L}}$ and map the solution back to the original lattice $\mathcal{L}$.

\section{Classical Primitives}\label{sec:classical}

\subsection{The LLL Algorithm}\label{sec:lll}

These lattice problems had been studied long before they became relevant to cryptography after Ajtai's worst-case/average-case equivalence result for lattices \cite{ajtai}. In fact, Gauss already provided a polynomial-time algorithm to solve SVP \textit{exactly} in 2-dimensions in the 19th century. Broadly, the algorithm finds the most orthogonal basis we can find for the given lattice using orthogonal projections, in an iterative manner that resembles Euclid's algorithm~\cite{ch17}. 

While solving SVP exactly (i.e., for $\gamma(n) = 1$) does not carry over to higher dimensions, a similar idea prevails. Namely, we would like to obtain a short basis, which is achieved through orthogonalizing such basis as much as possible. Suppose a basis $\{v_1, \ldots, v_n\}$ is orthogonal. The shortest vector in the lattice can be written as some linear combination $\lambda_1 = a_1 v_1 + \cdots + a_n v_n$ for $a_i \in \Z$, which has norm-squared $\| \lambda_1 \|^2 = a_1^2 \|v_1\|^2 + \cdots + a_n^2 \|v_n\|^2$. This is minimized for $a_i = \delta_{j0}$, such that $v_j$ is the shortest basis vector. We conclude that any orthogonal basis must include the shortest vector in the lattice. While not every lattice has an orthogonal basis, this reasoning explains why the LLL algorithm uses the Gram-Schmidt orthogonalization process as a subroutine. Throughout this paper, we will use the superscript $*$ to denote a Gram-Schmidt vector.

The LLL algorithm was introduced by Arjen Lenstra, Hendrik Lenstra, and László Lovász in 1982 \cite{lll}, originally in the context of polynomial factorization. The LLL algorithm reduces an arbitrary basis for a lattice into a ``shorter'' basis, which is referred to as the \textit{LLL-reduced basis} and is formally defined as follows:

\begin{definition}[$\delta$-LLL reduced basis] \label{LLL reduced}
    Let $\{b_1, b_2, \ldots, b_n\}$ be a basis for an $n$-dimensional lattice $\mathcal{L}$, and let $\{b_1^*, b_2^*, \ldots, b_n^*\}$ be the orthogonal basis generated with Gram-Schmidt. Let $\mu_{i, k} = \frac{\langle b_k, b_i \rangle}{\langle b_i^* \cdot b_i^* \rangle}$, where $\langle \cdot, \cdot \rangle$ denotes the dot product between two vectors. We say that $\{b_1, b_2, \ldots, b_n\}$ is a $\delta$-\textit{LLL-reduced basis} if it satisfies the following two conditions:
    \begin{enumerate}
        \item (Size reduced.) For all $i \neq k, \mu_{i, k} \leq 1/2$.
        \item (Lovász condition.) For each $i$, $||b_{i+1}^* + \mu_{i, i}b_i^*||^2 \geq \delta ||b_i^*||^2$. Equivalently, we can write $||b^{*}_{i+1}||^2 \geq (\delta - \mu^2_{i+1, i}) ||b_i^*||^2$.
    \end{enumerate}
\end{definition}

We remark that, in \cite{dvw21v2}, their usages of $\delta$ denote what we have defined as $1/\sqrt{\delta}$. That is, if $\delta'$ denotes the parameter used by \cite{dvw21v2}, then $\delta' = 1/\sqrt{\delta}$, where $\delta$ is defined as above. It is customary in the literature to choose $\delta = 3/4$, although any value $\delta \in (1/4, 1)$ guarantees that the LLL algorithm terminates in polynomial time. 

Put into words, the second condition reads as saying that $b^*_i$ is not much larger than $b^*_{i-1}$. Let us further develop the Lovász condition:
\[
    ||b^*_{i+1}||^2 \geq (\delta - \mu^2_{i+1, i}) ||b^*_i||^2  \geq \delta ||b_i^*||^2,
\]
\[
    ||b^*_i|| \geq \sqrt{\delta} ||b^*_i||.
\]
If we take logarithms in the Lovász condition, we obtain that
\begin{equation}
    \ln||b^*_{i+1}|| \geq  \ln(\sqrt{\delta} \cdot ||b^*_i||) = \ln ||b^*_i|| - \ln(1/\sqrt{\delta}).
\end{equation}
In notation of the Ducas and van Woerden note, where $\ell_i := \ln ||b^*_i||$, this corresponds to
\begin{equation}\label{eq:1}
    \ell_{i+1} \geq \ell_i - \ln(1/\sqrt{\delta}).
\end{equation}

\noindent If we re-state Equation~\ref{eq:1} in terms of the $\delta'$ used in the Ducas and van Woerden note, we obtain that
\begin{equation}\label{eq:2.3}
    \ell_{i+1} \geq \ell_i - \ln(\delta').
\end{equation}
This corresponds exactly to Fact 2.2 in the Ducas and van Woerden note.

It is natural to ask whether every lattice has an LLL-reduced basis. The LLL algorithm gives a constructive proof that the answer is yes. The LLL pseudocode is presented in Algorithm~\ref{alg:main1}.

\begin{algorithm}\caption{The LLL algorithm}\label{alg:main1}
\begin{algorithmic}[1]
\Procedure{\textsc{LLLAlg}}{$\{b_1, b_2, \ldots, b_n\}$}

    \State Compute the Gram-Schmidt basis $b^*_1, \ldots, b^*_n$ and coefficients $\mu_{i, j}$ for $1 \leq j < i \leq n$.
    \State Repeat the following two steps until the LLL reduced basis is found:
    \For{$i=1$ \textbf{to} $n$} \Comment{Perform size reduction}
        \For{$k = i-1$ \textbf{to} 1}
                \State $m \leftarrow \lfloor \mu_{k, i} \rfloor$
                \State $b_i \leftarrow b_i - mb_k$
        \EndFor
    \EndFor
    
    \For{$i=1$ \textbf{to} $n-1$} \Comment{Check Lovász condition}
        \If{$||b_{i+1}^* + \mu_{i, i+1}b^*_i||^2 < \delta||b_i^*||^2$}
        
            \State Swap $b_{i+1}$ and $b_i$.
            \State Go to Step 1.
        
        \EndIf
    \EndFor
\EndProcedure
\end{algorithmic}
\end{algorithm}

The correctness of the algorithm can be stated as follows \cite{ch17}:

\begin{lemma}
    If the LLL algorithm terminates, then the output basis is LLL-reduced. Moreover, for any $\delta \in (1/4, 1)$, the LLL algorithm terminates within polynomial time.
\end{lemma}

We will focus on the proof that the LLL algorithm terminates, because such a proof uses facts that we will need when showing correctness of the classical algorithm in Ducas and van Woerden. To show that the algorithm performs a finite number of steps, we define the following \textit{potential function} on the basis vectors of the lattice \cite{regev}:
\begin{equation}
    \phi(\textbf{B}) = \prod_{i=1}^n \det(\mathcal{L}_i) = \prod^n_{i=1} ||b^*_1|| \cdot \ldots \cdot ||b^*_i|| = \prod_{i=1}^n ||b^*_i||^{n-i+1},
\end{equation}
where $\mathcal{L}_i$ is defined as the lattice spanned by $b_1, \ldots, b_i$. Then, to show that the LLL algorithm terminates, we need to prove that the initial value of $\phi(\textbf{B})$ is upper bounded, that it can never drop below 1, and that each iteration of the LLL algorithm decreases $\phi$ by a factor of at least $\sqrt{4/3}$ \cite{peikert}.

The potential of the original basis $\textbf{B}$ is bounded by $\prod^n_{i=1} ||b_i||^n$, given that $||b^*_i|| \leq ||b_i||$. Because every intermediate basis is integral by definition and has positive integer determinant, the same holds for all the lattices $\mathcal{L}_i$ associated with that basis. Hence it follows that the potential is at least 1.

Let us consider the step \textit{Check Lov\'asz condition} in the pseudocode of the LLL algorithm (namely, Line 10). Suppose that $b_i$ and $b_{i+1}$ are swapped in Line 12 of the LLL algorithm pseudocode, and let the resulting basis be denoted $\textbf{B}'$. Then, we claim the following \cite{peikert}:

\begin{lemma}
If $b_i$ and $b_{i+1}$ are swapped in Line 12 of the LLL algorithm, then $b'^{*}_j = b^*_j$ for all $j \notin \{i, i+1\}$, and $b'^{*}_i = \mu_{i, i+1} b^*_i + b'_{i+1}$.
\end{lemma}

\begin{proof}
    If $j < i$, then the vector $b'^*_j$ is not affected by the swap, because it is the component of $b'_j = b_j$ orthogonal to $\text{span}(b'_1, \ldots, b'_{j-1}) = \text{span}(b_1, \ldots, b_{j-1})$. The same argument holds for $j > i+1$. Lastly, $b'^*_i$ is by definition the component of $b'_i = b_{i+1}$ orthogonal to $\text{span}(b'_1, \dots, b'_{i+1}) = \text{span}(b_1, \ldots, b_{i-1})$, which is $\mu_{i, i+1} b^*_i + b^*_{i+1}$ by construction.
\end{proof}

We can then prove the following \cite{peikert}:
\begin{lemma}
        If $b_i$ and $b_{i+1}$ are swapped in Line 12 of the LLL algorithm, then 
        \[
            \dfrac{\phi(\textbf{B}')}{\phi(\textbf{B})} < \sqrt{\delta}.
        \]
\end{lemma}

\begin{proof}
This immediately follows from the previous lemma. Define $\mathcal{L}_i$ as before, and let $\mathcal{L}'_i$ denote the lattice generated by the first $i$ basis vectors after a swap. Then:
\[
    \dfrac{\phi(\textbf{B}')}{\phi(\textbf{B})} = \dfrac{\det(\mathcal{L'}_i)}{\det(\mathcal{L}_i)} = \frac{||b^*_1|| \cdot \ldots \cdot ||b^*_{i-1}|| \cdot ||\mu_{i, i+1} b^*_i + b^*_{i+1}||}{||b^*_1|| \cdot \ldots \cdot ||b^*_{i-1}|| \cdot ||b^*_i||} = \frac{||\mu_{i, i+1}b^*_i + b^*_{i+1}||}{||b^*_i||} < \sqrt{\delta},
\]
given that, as per Line 11 in the LLL algorithm (Alg.~\ref{alg:main1}), the vectors are swapped if and only if $||b^*_{i+1} + \mu_{i, i+1} b^*_i||^2 < \delta ||b_i^*||^2$.
\end{proof}

Since $\delta \in (1/4, 1)$, the previous lemma in particular implies that
\[
    \prod_{j=i}^n \det(\mathcal{L'}_j) \leq \prod_{j=i}^n \det(\mathcal{L}_j).
\]
By taking logarithms, and by switching to the notation of Ducas and van Woerden, it follows that
\begin{equation}\label{partialvols}
    \sum^i_{j=1} \ell_j \leq \sum^i_{j=1} \ell'_j,
\end{equation}
which corresponds to their Fact 2.3. Note that in Ducas and van Woerden's notation, $\ell'_j$ is the logarithm of the norm of a vector from the original basis, and $\ell_j$ is the logarithm of the norm of a vector from the LLL-reduced basis.

Lastly, we ask, how does the LLL-reduced basis relate to the SVP? The following theorem shows their relationship; namely, the first vector in the ordered LLL-reduced basis corresponds to the desired approximation of $\lambda_1(\mathcal{L})$:
\begin{proposition}[$\gamma$-SVP from LLL]
    If $\{b_1, b_2, \ldots, b_n\}$ is an $n$-dimensional LLL-reduced basis of a lattice $\mathcal{L}$, then $||b_1|| \leq 2^{(n-1)/2} \lambda_1(\mathcal{L})$, where $\lambda_1(\mathcal{L})$ is the length of the shortest vector of $\mathcal{L}$.
\end{proposition}

The proof is short and can be consulted, for example, in Deng's manuscript \cite{deng}.

\subsection{Babai's Nearest Plane Algorithm}\label{sec:babai}

After explaining how the LLL algorithm provides an approximation to SVP, let us mention briefly how to obtain an approximation to CVP.

Babai's nearest plane algorithm solves $\gamma$-CVP for $\gamma(n) = 2^{n/2}$ in polynomial time. Because we do not need the inner mechanisms of Babai's algorithm in order to analyze the classical counterpart to Eldar and Hallgren's algorithm, other than its final approximation factor for CVP, we will not describe the algorithm in detail and instead only provide a short intuitive overview.

Babai's algorithm considers the affine subspace $H_c = \{ \sum_{i=1}^{n-1} a_i b_i + c b_n \}$. Then, $H_c = H_0 + cb_n$. However, $b_n$ is not necessarily orthogonal to $H_0$, and hence the distance between $H_0$ and $H_c$ is exactly $c ||b^*_n||$ (as opposed to $c||b_n||$). Babai then observes that inside each of these hyperplanes, there is a copy of the sublattice $\mathcal{L}'$ generated by the first $n-1$ basis vectors $b_i$ and shifted by $cb_n$. Then, Babai's algorithm iteratively finds the closest hyperplane $H_c$ to the target vector $t$ (as defined in $\gamma$-CVP) fixing one coordinate at a time \cite{babai}. 

When the input lattice to Babai's algorithm is an LLL-reduced basis (as described in Section~\ref{sec:lll}), the desired approximation for CVP follows. A proof of this claim can be found, for example, in \cite{noahcvp} or \cite{regevcvp}.

\section{Proof of Classical Alternative}\label{sec:proof}

Eldar and Hallgren's proof is interesting and uses ``phased cube states'' in a unique way, but their construction can be simplified. For the class of lattices on which they solve BDD, there in fact exists a deterministic, far simpler classical algorithm achieving the same result. Within three days of Hallgren's talk \cite{h21}, Ducas and van Woerden published a note \cite{dvw21} solving BDD with the appropriate subexponential factor for the special case $q = 2^n, k = 1$. 

Emboldened by this result, we sought to generalize their proof to the case of arbitrary $q, k$, demonstrating a classical alternative to Eldar and Hallgren's algorithm in all cases. However, on October 14, 2021, Ducas and van Woerden managed to prove the generalization we were working on $\cite{dvw21v2}$. Still, their proof is rather brief and leaves out many key details and intuitions. Here, we present a complete and self-contained proof of Theorem \ref{OHAyanishey}.

We remark that we already proved Facts 2.2 and 2.3 in \cite{dvw21v2} when describing the LLL algorithm and its correctness proof in Section~\ref{sec:lll}. We will use them again when finalizing the proof of the result in Section~\ref{sec:lastfacts}.

\subsection{Short Bases from Short Full-Rank Sets}\label{sec:51}
We open the proof with a lemma which shows that if we can find a full-rank set of short lattice vectors, then we can construct a lattice basis of short vectors, as well. Formally:
\begin{lemma}[Lemma 7.1 in \cite{MicciancioGoldwasser}, Fact 2.1 in \cite{dvw21v2}]
There exists a deterministic poly-time algorithm which, given a basis $\mathbf{B'}$ of an $n$-dimensional lattice $\mathcal{L}(\mathbf{B'})$, and a full-rank set of $n$ vectors $S \subset \mathcal{L}(\mathbf{B'})$ outputs a basis $\mathbf{B}$ such that $\|b_i^*\| \le \max_{s \in S} \|s\|$.
\end{lemma}\label{shortbasis}
We will make use of two subclaims regarding the relationships between unimodular and upper-triangular matrices. Unimodular matrices are defined below:
\begin{definition}
An integer matrix $U \in \Z^{n \times n}$ is \emph{unimodular} if $\det(U) = \pm 1$. Note that any unimodular matrix has integer inverse by Cramer's Theorem.
\end{definition}
We prove the following lemma:
\begin{lemma} \label{up-tri uni}
Given an integer matrix $Q \in \Z^{n \times n}$, there exists a unimodular matrix $U \in \Z^{n \times n}$ such that $T = UQ$ is upper-triangular.
\end{lemma}
\begin{proof}
$U$ will be constructed as a product of unimodular matrices corresponding to elementary matrices for row operations on an integer matrix: swapping two rows, multiplying a row by $-1$, and adding an integer multiple of one row to another. All such elementary matrices have determinant $\pm 1$, so by the homomorphism property of the determinant, so will $U$. We proceed by induction: if we can show that given a list $\{a_1, \ldots, a_k\}$ of integers (which we think of as the non-zero entries of a column of some $U'$ obtained from $U$ via a preceding sequence of elementary row operations), we can produce a new list $\{0, a'_1, \ldots, a_{k-1}' \}$ of integers by swapping entries in the list, multiplying entries by $-1$, and adding an integer multiple of one entry to another, then the proof will follow. First, multiply the appropriate entries by $-1$ such that all entries are non-negative. Then, order the entries by increasing value. If the first entry is 0, we are done; otherwise, subtract the maximal integer multiple of the first entry from all other entries which preserves non-negativity, and reorder in increasing value. The new first entry must have decreased in value from the previous first entry. This procedure is guaranteed to terminate in a finite number of steps, such that the first entry in eventually zero, which completes the proof.
\end{proof}
In addition, we will need the following lemma:
\begin{lemma} \label{triangular GS}
Suppose $S = BT$ for $T$ upper-triangular and invertible. Then $s_i^* =T_{ii} b_i^*$ for $s_j, b_j$ the $j$th columns of $S$ and $B$ respectively.
\end{lemma}
\begin{proof}
We proceed by induction. Suppose the claim holds for all $j \in [i-1]$. Then we have
$$    
s_i^* = s_i - \sum_{j=1}^{i-1} \frac{\langle s_j^*, s_i \rangle}{\langle s_j^*, s_j^* \rangle} s_j^* = s_i - \sum_{j=1}^{i-1} \frac{\langle b_j^*, s_i \rangle}{\langle b_j^*, b_j^* \rangle} b_j^*.
$$
Note that $s_i$ can be written as $s_i = \sum_{k=1}^i T_{ki} b_k$ since $T$ is upper-triangular, so
$$
s_i^* = \sum_{k=1}^i T_{ki} b_k - \sum_{j=1}^{i-1} \frac{\langle b_j^*, \sum_{k=1}^i T_{ki} b_k \rangle}{\langle b_j^*, b_j^* \rangle} b_j^* = \sum_{k=1}^i T_{ki} b_k - \sum_{k=1}^i  \sum_{j=1}^{i-1} \frac{\langle b_j^*, T_{ki} b_k \rangle}{\langle b_j^*, b_j^* \rangle} b_j^*.
$$
Consider just the term
$$
\sum_{k=1}^{i-1}  \sum_{j=1}^{i-1} \frac{\langle b_j^*, T_{ki} b_k \rangle}{\langle b_j^*, b_j^* \rangle} b_j^*.
$$
Since $b_k \in \text{span} \{b_1^*, \cdots, b_{i-1}^* \}$ for all $k \in [i-1]$, it follows that
$$
\sum_{j=1}^{i-1} \frac{\langle b_j^*, T_{ki} b_k \rangle}{\langle b_j^*, b_j^* \rangle} b_j^* = T_{ki} b_k.
$$
Therefore, we have
$$
s_i^* = T_{ii} b_i - T_{ii} \sum_{j=1}^{i-1} \frac{\langle b_j^*, b_i \rangle}{\langle b_j^*, b_j^* \rangle} b_j^* \equiv T_{ii} b_i^*,
$$
as desired.
\end{proof}
Armed with these intermediate results, we may now return to the proof of the main lemma:
\begin{proof}
Since $S$ is a set of lattice vectors, we can write $S = \mathbf{B}' Q$ for some integer matrix $Q \in \Z^{n \times n}$, which is invertible since $S$ is full rank. By Lemma~\ref{up-tri uni}, we may write $T = UQ$ for some unimodular $U$ such that $T$ is upper-triangular. Let $\mathbf{B} = \mathbf{B}' U^{-1}$. $\mathbf{B}$ is a basis since $U$ is unimodular, and therefore $U^{-1}$ is an integer matrix. Furthermore, $S = \mathbf{B} T$, so by Lemma~\ref{triangular GS} we have $\|b_i^*\| = \frac{1}{|T_{ii}|} \|s_i^*\| \le \|s_i^*\|$, since $T_{ii}$ are all integers. Since $\|s_i^*\| \le \|s_i\| \le \text{max}_{s \in S} \|s\|$, the proof is complete.
\end{proof}

\subsection{Duality}\label{sec:52}
In order to complete our proof, we will be using the idea of dual lattices, the definition of which is given as:
\begin{definition}
The dual of a lattice $\mathcal{L}$ is
\begin{align}
    \hat{\mathcal{L}} = \{v\in \text{span}(\mathcal{L}) | \langle v,w\rangle \in \mathbb{Z}\text{ for all }w\in\mathcal{L}\}.
\end{align}
\end{definition}
Equivalently, the dual lattice $\hat{\mathcal{L}}$ can be thought of as the set of all linear functions $\mathcal{L} \rightarrow \mathbb{Z}$. In order to use duality in our proof, we will first need to prove several lemmas, some of which are inspired by the exercises in \cite{micciancio2} by Daniele Micciancio.
\begin{lemma} \label{zn}
    The dual of $\mathbb{Z}^n$ is  $\mathbb{Z}^n$.
\end{lemma}
\begin{proof}
Take any $v \in \mathbb{Z}^n$. Since integers are closed under addition and multiplication, we have $\langle v,w \rangle \in \mathbb{Z}$ for any $w \in \mathbb{Z}^n$. This shows $\mathbb{Z}^n \subset \hat{\mathbb{Z}}^n$. Similarly, take any $v \in  \hat{\mathbb{Z}}^n$. For any $i\in \{1,\ldots,n\}$ we have $e_i \in \mathbb{Z}^n$ (where ${e}_i$ is the unit vector with 1 for the $i$th coordinate and 0 for the remaining coordinates). By definition of the dual lattice, this implies $\langle {v},{e}_i\rangle =v_i \in \mathbb{Z}$. Hence, ${v} \in \mathbb{Z}^n$ and $\hat{\mathbb{Z}}^n \subset \mathbb{Z}^n$. This proves that the dual of $\mathbb{Z}^n$ is $\mathbb{Z}^n$.
\end{proof}
\begin{lemma} \label{qzn}
    Let $\hat{\mathcal{L}}$ be the dual lattice of $\mathcal{L}$. For any $d>0$, the dual lattice of $d \mathcal{L}$ is $\frac{1}{d} \hat{\mathcal{L}}$.
\end{lemma}
\begin{proof} 
Take any $v \in \hat{(d\mathcal{L})}$. For all $w \in \mathcal{L}$, we have $d w \in d\mathcal{L}$, and hence:
\begin{align}
    \langle {v}, d {w}\rangle \in \mathbb{Z}  \Rightarrow \langle d v, {w}\rangle \in \mathbb{Z} \Rightarrow d {v} \in \hat{\mathcal{L}} \Rightarrow {v} \in \frac{1}{d} \hat{\mathcal{L}} \Rightarrow \hat{(d\mathcal{L})} \subset \frac{1}{d} \hat{\mathcal{L}}.
\end{align}
Similarly, take any  ${v}\in \frac{1}{d} \hat{\mathcal{L}}$, implying $d {v} \in  \hat{\mathcal{L}}$. Note that for all ${w} \in d \mathcal{L}$, we have $\frac{{w}}{d} \in \mathcal{L}$. Hence:
\begin{align}
   \frac{{w}}{d} \in \mathcal{L}, d{v} \in  \hat{\mathcal{L}} \Rightarrow  \left\langle d{v}, \frac{{w}}{d} \right\rangle \in \mathbb{Z} \Rightarrow  \langle {v}, {w}\rangle \in \mathbb{Z} \Rightarrow {v} \in \hat{(d\mathcal{L})} \Rightarrow \frac{1}{d} \hat{\mathcal{L}} \subset \hat{(d\mathcal{L})}  
\end{align}
Hence, $\frac{1}{d} \hat{\mathcal{L}}$ is the dual of $d\mathcal{L}$.
\end{proof}
\begin{theorem}[Theorem 2 from \cite{micciancio2}] \label{emintheo}
The dual of a lattice with basis $\mathbf{B}$ is a lattice with basis (the \emph{dual basis}) $\mathbf{D}=\mathbf{BG}^{-1}$, where $\mathbf{G}=\mathbf{B}^\top \mathbf{B}$ is the Gram matrix of $\mathbf{B}$. 
\end{theorem}
\begin{proof}
For any basis $\textbf{B}$, define $\textbf{D}$ as given in the theorem. The invertibility of \textbf{G} of follows from the linear independence of the columns of \textbf{B} (hence $\textbf{B}^\top \textbf{B}$ is nonsingular). Say $\mathcal{L}(\textbf{B})$ and $\mathcal{L}(\textbf{D})$ are the lattices corresponding to bases $\mathbf{B}$ and $\mathbf{D}$, respectively, and $\hat{\mathcal{L}}(\textbf{B})$ and $\hat{\mathcal{L}}(\textbf{D})$ are their dual lattices. From Definition \ref{deneme}, it follows that any vector in $\mathcal{L}(\mathbf{D})$ is of form $\mathbf{D}y$ for some integer vector $y$. Accordingly, take any $\mathbf{D}y \in \mathcal{L}(\mathbf{D})$. We have:
\begin{itemize}
    \item $\mathbf{Dy}= \mathbf{BG}^{-1}{y}= \mathbf{B}((\mathbf{B}^\top \mathbf{B})^{-1} {y}) \in \text{span}(\mathbf{B})$.
    \item For all $ \textbf{B}x \in \mathcal{L}(\mathbf{B})$, we have
    $\langle \mathbf{D}y, \mathbf{B}x \rangle=  (\mathbf{B}x)^\top (\mathbf{D}y)= {x}^\top \mathbf{B}^\top \mathbf{B}(\mathbf{B}^\top \mathbf{B})^{-1} {y} = {x}^\top {y} \in \mathbb{Z}$.
\end{itemize}
Hence, we have $\mathbf{D}y \in \hat{\mathcal{L}} (\textbf{B})$, implying $\mathcal{L}(\mathbf{D}) \subset \hat{\mathcal{L}}(\mathbf{B})$. Now, take any ${v} \in \hat{\mathcal{L}}(\mathbf{B})$. By definition of the dual, we must have ${v} \in \text{span}(B)$, implying that ${v}= \mathbf{B}w$ for some ${w} \in \mathbb{R}^n$. Moreover, since the columns of $\mathbf{B}$ are in $\mathcal{L}(\mathbf{B})$, it follows that $\mathbf{B}^\top {v} \in \mathbb{Z}^k$. Thus:
\begin{align}
    {v} = \textbf{B}w= \textbf{B} (\textbf{B}^\top \textbf{B})^{-1} \textbf{B}^\top \textbf{B} {w} =\textbf{D}(\textbf{B}^\top v) \in \mathcal{L}(\mathbf{D}).
\end{align}
This implies $\hat{\mathcal{L}} (\textbf{B})\subset \mathcal{L}(\mathbf{D})$. Therefore,  $\mathcal{L}(\mathbf{D})= \hat{\mathcal{L}} (\textbf{B})$, so $\mathbf{D}$ is indeed the dual basis of $\mathbf{B}$.
\end{proof}
\begin{lemma} \label{abidik}
    If $\mathbf{D}$ is the dual basis of  $\mathbf{B}$, then  $\mathbf{B}$ is the dual basis of $\mathbf{D}$. That is, the relationship between primal and dual bases is symmetric.
\end{lemma}
\begin{proof}
    Say $\textbf{D}$ is the dual basis of $\textbf{B}$. By Theorem \ref{emintheo}, this implies
        $\mathbf{D}=\textbf{B}(\textbf{B}^\top\textbf{B})^{-1}$. Hence we have
        \begin{align}
            \mathbf{D}^\top \mathbf{D}=(\textbf{B}^\top\textbf{B})^{-\top} \textbf{B}^\top\textbf{B}(\textbf{B}^\top\textbf{B})^{-1}=(\textbf{B}^\top\textbf{B})^{-\top} =   ((\textbf{B}^\top\textbf{B})^{\top})^{-1}=    (\textbf{B}^\top\textbf{B})^{-1}.   \label{bitir}\end{align}
            By Theorem \ref{emintheo} and (\ref{bitir}), the dual basis of $\mathbf{D}$ is:
\begin{align}
    \mathbf{D}(\textbf{D}^\top\textbf{D})^{-1}=\mathbf{D}((\textbf{B}^\top\textbf{B})^{-1}  )^{-1}=\mathbf{D}(\textbf{B}^\top\textbf{B})= \textbf{B}(\textbf{B}^\top\textbf{B})^{-1}(\textbf{B}^\top\textbf{B})=\textbf{B}.
\end{align}
Hence, the relation between primal and dual basis is indeed symmetric.
\end{proof}

\begin{lemma} \label{smith}
    For every lattice $\mathcal{L}$, we have $\det(\hat{\mathcal{L}})=\dfrac{1}{\det\left(\mathcal{L} \right)}$.
\end{lemma}
\begin{proof} Let $\mathbf{B}$ be the basis for $\mathcal{L}$ and $\mathbf{D}$ be the dual basis. Then by Theorem \ref{emintheo} we have  $\mathbf{B}^\top\mathbf{D}=\mathbf{B}^\top
\textbf{B}(\textbf{B}^\top\textbf{B})^{-1}=\textbf{I}$. By Lemma \ref{abidik}, it follows that $\mathbf{D}^\top\mathbf{B}=\mathbf{I}$. This implies that
\begin{align}
    (\mathbf{B}^\top \mathbf{B})  (\mathbf{D}^\top \mathbf{D})=  (\mathbf{B}^\top (\mathbf{B}\mathbf{D}^\top )\mathbf{D})= (\mathbf{B}^\top \mathbf{D})= \mathbf{I}.
\end{align}
Hence, the Gram matrix of $\mathbf{B}$ is the inverse of the Gram matrix of $\mathbf{D}$. This implies that $\det(\mathbf{B}^\top \mathbf{B})= 1 / \det(\mathbf{D}^\top \mathbf{D})$ and therefore $\det(\mathcal{L})=1/\det(\mathcal{\hat{L}})$. 
\end{proof}
We now use the results above for our main proof. Recall that the lattice we are considering is $L_{\mathbf{A}} = q\mathbb{Z}^n+\mathbf{A} \mathbb{Z}^k$. It follows from Lemmas \ref{zn} and \ref{qzn} that the dual of $q\mathbb{Z}^n$ is $\frac{1}{q}\mathbb{Z}^n$. Since $ q\mathbb{Z}^n \subset q\mathbb{Z}^n+\mathbf{A} \mathbb{Z}^k$, it follows that the dual of $L_{\mathbf{A}}$ must be contained in $\frac{1}{q}\mathbb{Z}^n$. This is because adding more vectors to the lattice simply adds more constraints to the dual lattice. In other words, any vector ${w}$ in the dual of $L_{\mathbf{A}}$ needs to have $\langle {v},{w} \rangle \in \Z $ for all ${v} \in q\mathbb{Z}^n \subset L_{\mathbf{A}}$, implying that ${w} \in \frac{1}{q} \mathbb{Z}^n$. Let $\textbf{B}=\{b_1,..,b_n\}$ be a LLL-reduced basis for $L_{\mathbf{A}}$, and  and let $\textbf{D}=\{d_1,...,d_n\}$ be the dual basis for $\textbf{B}$, as defined in Theorem $\ref{emintheo}$. Lastly, let $\textbf{B}^*=\{b^*_1,..,b^*_n\}$ and $\textbf{D}^*=\{d^*_1,..,d^*_n\}$ be the orthogonal bases generated by applying Gram-Schmidt to $\textbf{B}$ and $\textbf{D}$, respectively. Since the dual lattice is contained in $\frac{1}{q}\mathbb{Z}^n$, it follows that \begin{align}
\det(\hat{\mathcal{L}}_{i+1})=\prod_{j=i+1}^n ||d^*_j|| \geq \frac{1}{q^{n-i}}. \label{flu}
\end{align} 
where $\hat{\mathcal{L}}_{i+1}$ is defined as the lattice spanned by $d_{i+1}, \ldots, d_{n}$. Using Lemma \ref{smith} and Corollary 6 from \cite{micciancio2} (which explains how the relationship in Lemma \ref{smith} extends to the orthogonalizations of dual bases), we also know that:
\begin{align}
    \prod_{j=i+1}^n ||b^*_j|| = \left(\prod_{j=i+1}^n ||d^*_j||  \right)^{-1}, \label{shot}
\end{align} 
since we can treat the last $n-i$ vectors as a partial basis. Combining (\ref{flu}) and (\ref{shot}), we get:
\begin{align}
\prod_{j=i+1}^n ||b^*_j|| &\leq q^{n-i}\\
\Rightarrow \ln\left(\prod_{j=i+1}^n ||b^*_j|| \right) &\leq \ln q^{n-i}\\
\Rightarrow \sum_{j=i+1}^n\ln||b^*_j||=\sum_{j=i+1}^n\ell_j &\leq (n-i)\ln q, \label{leftover_volume}
\end{align} 
which corresponds to Fact 2.4 from \cite{dvw21v2}.

\subsection{Completing the Proof}\label{sec:lastfacts}

Putting all of these pieces together, we have the following proposition, which corresponds to Proposition 2.5 in Ducas and van Woerden's note: 
\begin{proposition}
    $$\min_i \ell_i \ge \ln\big(\lambda_1(\mathcal{L})\big) - \sqrt{2k\cdot \ln \delta' \cdot \ln q},$$
where $\delta' = \frac{1}{\sqrt{\delta}}$ for $\delta$ the traditional parameter in the LLL algorithm, as explained immediately after Definition \ref{LLL reduced}.
\end{proposition} 
\begin{proof}
Let $d = \Big\lceil\sqrt{\frac{2k\cdot\ln q}{\ln \delta'}}\Big\rceil$, we can then rewrite the above inequality as 
$$\min_i \ell_i \ge \ln\big(\lambda_1(\mathcal{L})\big) - (d - 1)\ln(\delta'),$$
because $d - 1 \le \sqrt{\frac{2k\cdot\ln q}{\ln \delta'}}.$ Note that when $i \le d,$ the proposition follows from Equation~\ref{eq:2.3}. This is because $\ell_1$ is the magnitude of a lattice vector $\textbf{b}^*_1$, so $\ell_1 \ge \ln\big(\lambda_1(\mathcal{L})\big)$ by definition, and we can iteratively apply the Lov\'asz condition until we reach $\ell_1:$
\begin{align}
    \ell_i &\ge  \ell_{i - 1} - \ln(\delta') \\ 
    &\ge  \big(\ell_{i-2} - \ln(\delta') \big) - \ln(\delta') = \ell_{i-2} - 2\ln(\delta') \\ 
    & \vdots \\ 
    &\ge \ell_1 - (i - 1)\ln(\delta')\\ 
    & \ge \ln(\lambda_1(\mathcal{L})) - (i - 1)\ln(\delta').
\end{align}
It follows that for $i \le d,$ this is at least $\ln(\lambda_1(\mathcal{L})) - (d - 1)\ln(\delta')$.
\\ 
\\  
For the $i > d$ case, we can start by bounding the sum $\sum_{j=i}^i \ell_j$:
\begin{equation}\label{partialvols2}
    \sum_{j = 1}^i \ell_j \le \sum_{j = 1}^i \ell'_j \le i\ln(q).
\end{equation}
This follows directly from the fact that $\sum_{j = 1}^i \ell_j \le \sum_{j = 1}^i \ell'_j$ (Equation~\ref{partialvols}), and the fact that $\ell'_i \le \ln(q)$ for all $\ell'_i$ (Lemma~\ref{shortbasis}).
We can also establish a lower bound by subtracting Equation~\ref{leftover_volume} from \ref{partialvols2}: 
\begin{equation}
   (i - k)\ln(q) \le \sum_{j = 1}^i \ell_j \le i\ln(q) \quad\quad\textnormal{ for } d < i \le n. 
\end{equation}
Now, we use the fact that $i > d$ and subtract
\begin{align}
    \sum_{j = 1}^i \ell_j - \sum_{j = 1}^{i - d} \ell_j &\ge (i - k)\ln(q) - (i - d)\ln(q), \\ 
    \sum_{j = i - d + 1}^i \ell_j &\ge  (d - k)\ln(q).
\end{align}
Furthermore, we can iteratively apply the Lov\'asz condition like we did above to obtain $\ell_j \le \ell_i + (j-i)\ln(\delta')$ for $j \le i.$ Then, we can use this fact to establish an upper bound: 
\begin{equation}
(d - k)\ln(q) \le \sum_{j = i - d + 1}^i \ell_j \le  d\ell_i + \sum_{j = 0}^{d - 1} j\ln(\delta').
\end{equation}
Rewriting this, we have 
\begin{equation}
(d - k)\ln(q) \le  d\ell_i + \Big(\frac{d(d-1)}{2}\Big)\ln(\delta').
\end{equation}
Then, we divide everything by $d$ and rearrange to obtain 
\begin{align}
\ell_i &\ge \frac{(d - k)\ln(q)}{d} - \Bigg(\frac{(d-1)}{2}\Bigg)\ln(\delta'), \\ 
\ell_i &\ge \ln(q) - \frac{k}{d}\ln(q) - \Bigg(\frac{(d-1)}{2}\Bigg)\ln(\delta').
\end{align} 
Recall that we defined $d$ such that $d - 1 \le \sqrt{\frac{2k\cdot\ln q}{\ln \delta'}} \le d,$ so we can substitute and simplify: 
\begin{align}
\ell_i &\ge \ln(q) - \frac{k}{d}\ln(q) - (d-1)\frac{1}{2}\ln(\delta'), \\ 
\ell_i &\ge \ln(q) - \frac{k}{\sqrt{\frac{2k\cdot\ln q}{\ln \delta'}}}\ln(q) - \Bigg(\sqrt{\frac{2k\cdot\ln q}{\ln \delta'}}\Bigg) \frac{1}{2}\ln(\delta'), \\ 
\ell_i &\ge \ln(q) - \tfrac{1}{2}\sqrt{2k\ln(q)\cdot \ln(\delta')} - \frac{1}{2}\sqrt{2k\cdot\ln(q) \cdot \ln(\delta')}.
\end{align} 
Finally, we use the fact that $\lambda_1(\mathcal{L})$ must be less than $q$ to obtain 
\begin{align}
\ell_i &\ge \ln(\lambda_1(\mathcal{L})) - \sqrt{2k\ln(q)\cdot \ln(\delta')} \quad\quad \textnormal{ for } d < i \le n,
\end{align} 
which is what we wanted to show. 
\end{proof}

All that remains is to show that Babai's nearest plane algorithm will solve BDD. Recall that Babai's algorithm solves $\gamma$-CVP for $\gamma(n) \le 2^{n/2}$ \cite{regevcvp}, and when $r < \lambda_1(\mathcal{L}) / 2,$ then CVP solves BDD because the solution is unique. In other words, there is only one possible vector within $r$ in this case, and that is the closest one. Moreover, the smallest Gram-Schmidt vector $\min_i \|b_i^*\|$ is less than $\lambda_1(\mathcal{L})$ \cite{epfl}, and therefore Babai's algorithm will solve BDD for $r < \min_i \|b_i^*\|.$ Due to the exponential factor and the above proposition, we have shown that Babai's algorithm can solve BDD up to radius $\tfrac{1}{2}\lambda_1(\mathcal{L})\cdot \exp{-\sqrt{2k\cdot\ln(q)\cdot\ln(\delta')}}$, concluding the proof of Theorem~\ref{OHAyanishey}.

\section{Conclusion and Future Work}

After a survey of numerous interrelated lattice problems and previous work on quantum approaches to these problems, we have presented an overview of Eldar and Hallgren's work, as well as a complete proof of Ducas and van Woerden's result with all details filled in.

There is further work to be done in the area of classical algorithms for BDD and related lattice problems. In particular, instead of LLL, it would be worthwhile to consider alternative basis reductions, such as a Block Korkin-Zolotarev (BKZ) reduction, which is known to give a better approximation factor than LLL when solving $\gamma$-CVP using Babai's nearest plane algorithm \cite{book}. It would also be interesting to consider randomized heuristics for lattice problems, such as Klein's algorithm \cite{kleinalgo}.

There is a more general question of when a quantum algorithm for a lattice problem is guaranteed to have a classical counterpart. If more quantum-classical pairings are found, it would be interesting to consider the connections between these cases and conjecture such a result. 
\section*{Acknowledgements}

We would like to thank Professor Boaz Barak and Emil Khabiboulline for their guidance throughout this project. We are also thankful to Professor Vinod Vaikuntanathan and Jessica Sorrell for useful correspondence on lattice-based quantum algorithms and their classical counterparts, and to Professor Sean Hallgren and Dr. Lior Eldar for helpful feedback and for kindly providing us with the manuscript of their result.

\bibliography{refs.bib}

\begin{thebibliography}{DvW21b}

\bibitem[Ajt98]{ajtai}
Mikl{\'o}s Ajtai.
\newblock The shortest vector problem in {L}2 is {N}{P}-hard for randomized
  reductions.
\newblock In {\em Proceedings of the 13th annual ACM Symposium on Theory of
  Computing}, pages 10--19, 1998.

\bibitem[Bab86]{babai}
L{\'a}szl{\'o} Babai.
\newblock On {L}ov{\'a}sz’ lattice reduction and the nearest lattice point
  problem.
\newblock {\em Combinatorica}, 6(1):1--13, 1986.

\bibitem[Bai06]{epfl}
Thomas Baigneres.
\newblock Cryptosystems and {L}{L}{L}.
\newblock Technical report, 2006.

\bibitem[Bar22]{boaz}
Boaz Barak.
\newblock {\em An {I}ntensive {I}ntroduction to {C}ryptography}.
\newblock 2022.

\bibitem[DD18]{lucaslec}
Daniel Dadush and L{\'e}o Ducas.
\newblock Lecture 4 for the course ``{I}ntro to {L}attice {A}lgorithms and
  {C}ryptography''.
\newblock
  \url{https://homepages.cwi.nl/~dadush/teaching/lattices-2018/notes/lecture-4.pdf},
  2018.

\bibitem[Den16]{deng}
Xinyue Deng.
\newblock An introduction to {L}enstra-{L}enstra-{L}ovasz {L}attice basis
  reduction algorithm, 2016.

\bibitem[DvW21a]{dvw21}
L{\'e}o Ducas and Wessel van Woerden.
\newblock A note on a {C}laim of {E}ldar \& {H}allgren: {L}{L}{L} already
  solves it ({V}ersion 1).
\newblock {\em Cryptology ePrint Archive}, 2021.

\bibitem[DvW21b]{dvw21v2}
L{\'e}o Ducas and Wessel van Woerden.
\newblock A note on a {C}laim of {E}ldar \& {H}allgren: {L}{L}{L} already
  solves it ({V}ersion 2).
\newblock {\em Cryptology ePrint Archive}, 2021.

\bibitem[EH22]{eh21}
Lior Eldar and Sean Hallgren.
\newblock An efficient quantum algorithm for lattice problems achieving
  subexponential approximation factor.
\newblock {\em arXiv preprint arXiv:2201.13450}, 2022.

\bibitem[Gal12]{ch17}
Steven~D Galbraith.
\newblock {\em Mathematics of Public Key Cryptography}.
\newblock Cambridge University Press, 2012.

\bibitem[GMSS99]{gmss99}
Oded Goldreich, Daniele Micciancio, Shmuel Safra, and J-P Seifert.
\newblock Approximating shortest lattice vectors is not harder than
  approximating closest lattice vectors.
\newblock {\em Information Processing Letters}, 71(2):55--61, 1999.

\bibitem[Hal21]{h21}
Sean Hallgren.
\newblock An efficient quantum algorithm for lattice problems achieving
  subexponential approximation factor.
\newblock Available at \url{https://www.youtube.com/watch?v=K5Apl_qCnDA}, 2021.

\bibitem[Kup05]{copernicus}
Greg Kuperberg.
\newblock A subexponential-time quantum algorithm for the dihedral hidden
  subgroup problem.
\newblock {\em SIAM Journal on Computing}, 35(1):170--188, 2005.

\bibitem[Kup11]{copernicus2}
Greg Kuperberg.
\newblock Another subexponential-time quantum algorithm for the dihedral hidden
  subgroup problem.
\newblock {\em arXiv preprint arXiv:1112.3333}, 2011.

\bibitem[LLL82]{lll}
Arjen~K. Lenstra, Hendrik~Willem Lenstra, and L{\'a}szl{\'o} Lov{\'a}sz.
\newblock Factoring polynomials with rational coefficients.
\newblock {\em Mathematische Annalen}, 261(ARTICLE):515--534, 1982.

\bibitem[LLS11]{kleinalgo}
Shuiyin Liu, Cong Ling, and Damien Stehl{\'e}.
\newblock Decoding by sampling: A randomized lattice algorithm for bounded
  distance decoding.
\newblock {\em IEEE Transactions on Information Theory}, 57(9):5933--5945,
  2011.

\bibitem[LM09]{LM09}
Vadim Lyubashevsky and Daniele Micciancio.
\newblock On bounded distance decoding, unique shortest vectors, and the
  minimum distance problem.
\newblock In {\em Annual International Cryptology Conference}, pages 577--594.
  Springer, 2009.

\bibitem[MG02]{MicciancioGoldwasser}
Daniele Micciancio and Shafi Goldwasser.
\newblock {\em Complexity of Lattice Problems: A Cryptographic Perspective},
  volume 671.
\newblock Springer Science \& Business Media, 2002.

\bibitem[Mic01]{Micciancionphard}
Daniele Micciancio.
\newblock The shortest vector in a lattice is hard to approximate to within
  some constant.
\newblock {\em SIAM Journal on Computing}, 30(6):2008--2035, 2001.

\bibitem[Mic07]{micciancio2}
Daniele Micciancio.
\newblock Lecture 3 for the course ``{L}attice {A}lgorithms and
  {A}pplications''.
\newblock \url{https://cseweb.ucsd.edu/classes/sp07/cse206a/lec3.pdf}, 2007.

\bibitem[NC10]{mikeandike}
Michael~A Nielsen and Isaac Chuang.
\newblock {\em Quantum Computation and Quantum Information}.
\newblock Cambridge University Press, 2010.

\bibitem[NV10]{book}
Phong~Q. Nguyen and Brigitte Vall{\'e}e.
\newblock {\em The LLL algorithm}.
\newblock Springer, 2010.

\bibitem[Pei09]{Peikert09}
Chris Peikert.
\newblock Public-key cryptosystems from the worst-case shortest vector problem.
\newblock In {\em Proceedings of the 41st Annual ACM Symposium on Theory of
  Computing}, pages 333--342, 2009.

\bibitem[Pei13]{peikert}
Chris Peikert.
\newblock Lecture 3 for the course ``{L}attices in {C}ryptography''.
\newblock \url{https://web.eecs.umich.edu/~cpeikert/lic13/lec03.pdf}, 2013.

\bibitem[Reg04a]{regev}
Oded Regev.
\newblock Lecture 2 for the course ``{L}attices in {C}omputer {S}cience''.
\newblock
  \url{https://cims.nyu.edu/~regev/teaching/lattices_fall_2004/ln/lll.pdf},
  2004.

\bibitem[Reg04b]{regevcvp}
Oded Regev.
\newblock Lecture 3 for the course ``{L}attices in {C}omputer {S}cience''.
\newblock
  \url{https://cims.nyu.edu/~regev/teaching/lattices_fall_2004/ln/cvp.pdf},
  2004.

\bibitem[Reg04c]{reg04}
Oded Regev.
\newblock Quantum computation and lattice problems.
\newblock {\em SIAM Journal on Computing}, 33(3):738--760, 2004.

\bibitem[Reg04d]{reg06}
Oded Regev.
\newblock A subexponential time algorithm for the dihedral hidden subgroup
  problem with polynomial space.
\newblock {\em arXiv preprint quant-ph/0406151}, 2004.

\bibitem[Reg10]{reg05}
Oded Regev.
\newblock The learning with errors problem.
\newblock {\em Invited survey in CCC}, 7(30):11, 2010.

\bibitem[SD16]{noahcvp}
Noah Stephens-Davidowitz.
\newblock Lecture 5 for the course ``{L}attices {M}ini {C}ourse''.
\newblock \url{http://www.noahsd.com/mini_lattices/05__babai.pdf}, 2016.

\bibitem[SD21]{latticereds}
Noah Stephens-Davidowitz.
\newblock Dimension-preserving reductions between lattice problems, 2015.
\newblock 2021.

\bibitem[vEB81]{boas}
Peter van Emde~Boas.
\newblock Another {N}{P}-complete problem and the complexity of computing short
  vectors in a lattice.
\newblock {\em Technical Report, Department of Mathematics, University of
  Amsterdam}, 1981.

\end{thebibliography}

\end{document}